\documentclass[11pt]{article}

\usepackage[letterpaper,margin=1in]{geometry}
\usepackage[T1]{fontenc}
\usepackage[utf8]{inputenc}
\usepackage{lmodern}
\usepackage{microtype}
\usepackage{amsmath,amssymb,amsthm,mathtools}
\usepackage{booktabs,tabularx,array}
\usepackage{enumitem}
\usepackage{float}
\usepackage{xcolor}
\usepackage[numbers]{natbib}

\usepackage[normalem]{ulem}

\newif\ifshowcomments
\showcommentstrue
\ifshowcomments
  \usepackage[colorinlistoftodos,textsize=small]{todonotes}
\else
  \usepackage[disable]{todonotes}
\fi

\usepackage{xspace}
\usepackage{tikz}
\usetikzlibrary{arrows.meta,positioning,calc,fit,backgrounds,decorations.pathreplacing,shapes.geometric}
\usepackage{hyperref}
\usepackage{aliascnt}
\usepackage[nameinlink,noabbrev]{cleveref}

\definecolor{ravblue}{RGB}{25,74,120}
\definecolor{ravlight}{RGB}{243,247,251}
\definecolor{ravgreen}{RGB}{32,112,72}
\definecolor{ravorange}{RGB}{164,91,18}
\definecolor{ravgray}{RGB}{91,101,111}

\hypersetup{
  pdftitle={Anchoring for Truthfulness: The Random-Anchor Volume Mechanism for Multi-Facility Location},
  pdfauthor={Haris Aziz, Simon Mackenzie, Mashbat Suzuki},
  colorlinks=true,
  linkcolor=ravblue,
  citecolor=ravblue,
  urlcolor=ravblue
}

\newtheorem{mainthm}{Theorem}

\newcommand{\restatedhead}{}
\newtheorem*{restatedmainthmaux}{\restatedhead}
\newenvironment{restatedmainthm}[2]{%
  \renewcommand{\restatedhead}{Theorem~\ref{#1}}%
  \begin{restatedmainthmaux}[#2, restated]}{\end{restatedmainthmaux}}

\newtheorem{theorem}{Theorem}[section]
\newaliascnt{lemma}{theorem}
\newtheorem{lemma}[lemma]{Lemma}
\aliascntresetthe{lemma}
\newaliascnt{proposition}{theorem}
\newtheorem{proposition}[proposition]{Proposition}
\aliascntresetthe{proposition}
\newaliascnt{corollary}{theorem}

\aliascntresetthe{corollary}
\theoremstyle{definition}
\newaliascnt{definition}{theorem}
\newtheorem{definition}[definition]{Definition}
\aliascntresetthe{definition}
\newaliascnt{example}{theorem}
\newtheorem{example}[example]{Example}
\aliascntresetthe{example}
\theoremstyle{remark}
\newaliascnt{remark}{theorem}

\aliascntresetthe{remark}

\crefname{equation}{Equation}{Equations}
\crefname{figure}{Figure}{Figures}
\crefname{table}{Table}{Tables}
\crefname{section}{Section}{Sections}
\crefname{subsection}{Section}{Sections}
\crefname{mainthm}{Theorem}{Theorems}
\Crefname{mainthm}{Theorem}{Theorems}
\crefname{theorem}{Theorem}{Theorems}
\crefname{lemma}{Lemma}{Lemmas}
\crefname{proposition}{Proposition}{Propositions}
\crefname{corollary}{Corollary}{Corollaries}
\crefname{definition}{Definition}{Definitions}
\crefname{example}{Example}{Examples}
\crefname{remark}{Remark}{Remarks}

\newcommand{\RAV}{\textnormal{\textsc{RAV}}\xspace}
\newcommand{\RAVk}[1]{\texorpdfstring{\ensuremath{\textnormal{\textsc{RAV}}_{#1}}}{RAV-#1}\xspace}
\newcommand{\SC}{\operatorname{SC}}
\newcommand{\cost}{\operatorname{cost}}
\newcommand{\OPT}{\operatorname{OPT}}
\newcommand{\E}{\mathbb E}
\newcommand{\R}{\mathbb R}
\newcommand{\given}{\,\middle|\,}
\newcommand{\abs}[1]{\left|#1\right|}
\newcommand{\set}[1]{\left\{#1\right\}}

\ifshowcomments
  \newcommand{\haris}[1]{\textcolor{blue}{Haris says: #1}}
  \newcommand{\simon}[1]{\textcolor{red}{Simon says: #1}}

 \newcommand{\mash}[1]{\textcolor{teal}{Mashbat says: #1}}
\else
  \newcommand{\haris}[1]{}
  \newcommand{\simon}[1]{}
  \newcommand{\mash}[1]{}
\fi

\setlist[itemize]{topsep=4pt,itemsep=2pt,parsep=1pt,leftmargin=1.5em}
\setlist[enumerate]{topsep=4pt,itemsep=3pt,parsep=1pt,leftmargin=1.7em}
\tikzset{
  ravpoint/.style={circle,fill=ravblue,inner sep=2.2pt},
  ravanchor/.style={diamond,fill=ravorange,inner sep=2.5pt},
  ravbox/.style={draw=ravblue,rounded corners=2pt,fill=ravlight,
    align=center,text width=3.0cm,minimum height=1.1cm},
  ravarrow/.style={-{Latex[length=2mm]},thick,draw=ravblue},
  gapbrace/.style={decorate,decoration={brace,amplitude=4pt},thick}
}

\title{\textbf{Anchoring for Truthfulness: The Random-Anchor Volume Mechanism for Multi-Facility Location}}
\author{%
  Haris Aziz\\
  {\small UNSW Sydney}\\
  {\small\texttt{haris.aziz@unsw.edu.au}}
  \and
  Simon Mackenzie\\
  {\small UNSW Sydney}\\
  {\small\texttt{simon.william.mackenzie@gmail.com}}
  \and
  Mashbat Suzuki\\
  {\small UNSW Sydney}\\
  {\small\texttt{mashbat.suzuki@unsw.edu.au}}%
}
\date{}

\begin{document}
\maketitle
\begin{abstract}
We study the strategyproof placement of \(k\) facilities on the real line for
\(n\) agents who privately report their locations, without monetary transfers.
For two facilities, the Proportional Mechanism of Lu, Sun, Wang, and Zhu
(2010) is strategyproof in expectation and achieves a constant-factor
approximation to the optimal social cost. Whether such a guarantee is possible
for three facilities in the standard model, where each agent is served by her
nearest open facility, has remained open.

We resolve this question affirmatively by introducing the
\emph{Random-Anchor Volume} mechanism. The mechanism first opens a facility at
the report of a uniformly random agent, called the \emph{anchor}, and then
jointly selects two additional reports, assigning each pair probability
proportional to the product of the two consecutive gaps formed by the pair and
the anchor. We prove that the mechanism is strategyproof in expectation and
has expected social cost at most \(8\OPT_3\), where \(\OPT_k\) denotes the
minimum social cost achievable using at most \(k\) facilities.

The mechanism naturally extends to every \(k\geq 2\) by selecting \(k-1\)
additional reports with probability proportional to the product of the
consecutive gaps among them and the anchor. Under truthful reporting, this
generalization has expected social cost at most
\(4(k-1)\OPT_k\). Its incentive guarantee, however, has a sharp boundary: the
mechanism is strategyproof in expectation for \(k\in\{1,2,3\}\), but is
manipulable for every \(k\geq 4\).

\end{abstract}

\section*{Generative AI disclosure}

During the development of this work, the authors made extensive use of
OpenAI's ChatGPT (GPT-5.6, ``Sol'').  The Random-Anchor Volume mechanism
was discovered by the system, and the results of this paper were then also
proved by it, through an extended research interaction.  The human authors'
contribution was chiefly to ask the relevant questions and to turn the
resulting material into a paper: refining the arguments, imposing
coherence, and improving the writing.  All mathematical statements, proofs,
and literature claims were verified by the human authors, who take full
responsibility for their correctness, originality, and presentation.

\newpage

\tableofcontents

\section{Introduction}\label{sec:introduction}

Facility location lies at the intersection of operations research and collective
decision making \cite{Daskin2013,Black1958,Moulin1988,BrandtEtAl2016}.  In
operations research, locating warehouses, hospitals, and emergency services
is a canonical optimization problem; in economics and social choice, a
location on a line is also an ideal point, and choosing a facility is a public
decision under single-peaked preferences.

The strategic version introduces a challenge that is absent from the
classical optimization problem.  A public authority may want to locate
clinics, schools, charging stations, or other services near the people who
will use them, but those locations must often be elicited from the users
themselves.  Monetary transfers may be infeasible, legally unavailable, or
simply inappropriate for such public decisions.  The authority must then
reconcile two objectives: reporting a true location should be a dominant
strategy, and the chosen facilities should minimize the costs of the agents. 

This tension is particularly clean on the real line.  
We consider $n$ agents
and at most $k$ interchangeable facilities.  An agent's cost is her distance
from her true location to the nearest open facility, and social cost is the
sum of these distances.  A randomized mechanism is \emph{strategyproof in
expectation} if no agent can reduce her expected distance by changing only
her own report.  Thus a deviation may help in some realizations and hurt in
others, but it cannot help on average over the mechanism's random choices.
We measure efficiency by the worst-case ratio between expected social cost
and the optimum $k$-facility social cost.

From the viewpoint of computation, the striking feature is that the
underlying one-dimensional $k$-median problem is tractable: approximation is
needed here to overcome incentives, not computational hardness.  From the
viewpoint of economic theory, the same ratio gives a quantitative price of
dominant-strategy implementation without transfers.  This is precisely the
perspective of \emph{approximate mechanism design without money}, introduced
by Procaccia and Tennenholtz with facility location as its central case study
\cite{ProcacciaTennenholtz2013}.

The one-facility case suggests that the two objectives may coexist perfectly.
Opening at a median report minimizes social cost and is group-strategyproof.
The result belongs to the classical theory of majority rule under
single-peaked preferences developed by Black \cite{Black1948,Black1958}.
Moulin subsequently characterized anonymous deterministic strategyproof rules on the
line as generalized medians obtained by adding fixed phantom reports
\cite{Moulin1980,Moulin1988}.  
In contrast, when there are multiple facilities, there is an inherent tradeoff between strategyproofness and social cost minimization. For example, for two facilities, every deterministic strategyproof mechanism has an approximation ratio $\Omega(n)$~\citep{LuEtAl2010}. 
To achieve strategyproofness together with a sublinear approximation ratio, a natural recourse is randomization. Indeed, randomized mechanisms were also explored by Procaccia and Tennenholtz~\cite{ProcacciaTennenholtz2013} in their quest to achieve better approximation bounds. 

For two facilities, Lu, Sun, Wang,
and Zhu~\cite{LuEtAl2010} proposed the randomized Proportional Mechanism that first chooses an agent uniformly at random 
and then chooses the second facility-hosting agent with probability
proportional to her distance from the first. They proved that it is strategyproof in
expectation and achieves approximation factor $4$ in every metric space. For three facilities,  they also proposed a different  mechanism that is strategyproof (in expectation) and achieves an $O(n)$ approximation. 
They asked whether a constant approximation is possible with three
or more facilities~\cite[Secs.~4.3 and 6]{LuEtAl2010}.  The later journal version of Procaccia
and Tennenholtz's foundational paper summarized the difficulty by observing
that the intuition behind the positive two-facility results ``already
collapses'' at three facilities, and identified the extension to three and
beyond as a central open problem
\cite[Secs.~4.3 and 6]{ProcacciaTennenholtz2013}.  The 2021 survey of Chan et
al. still recorded only a population-dependent guarantee in the 
multi-facility model \cite[Sec.~2.3]{ChanEtAl2021}.

\subsection{Our contribution}

We present the first strategyproof (in expectation) randomized algorithm that places at most 3 facilities and guarantees a constant factor approximation to the optimal social cost. We call the mechanism \emph{Random-Anchor Volume (RAV)}. 
The mechanism first chooses one agent
uniformly and opens a facility at her report; this is the \emph{anchor}.  It then
considers every pair of other agents.  If the anchor and the two reports, in
sorted order, are $z_0\le z_1\le z_2$, the pair receives weight
$(z_1-z_0)(z_2-z_1).$
The pair is selected with probability proportional to this product, and the remaining facilities open at the two selected reports.

More generally, we define a family of Random-Anchor Volume mechanisms that apply to any target number of facilities. For each integer $k\geq 1$, there is a corresponding \RAVk{k} mechanism that places at most $k$ facilities. 
After fixing an anchor uniformly at random, \RAVk{k} assigns to each set of
$k-1$ further reports a weight equal to the product of the consecutive gaps
among the anchor and those reports. It then selects a set with probability proportional to its weight and opens facilities at the anchor and the selected reports.

Our first theorem identifies the exact truthfulness boundary of the \RAV
family.

\begin{mainthm}[Exact \RAV truthfulness boundary]\label{thm:intro-boundary}
For $k\in\{1,2,3\}$, the mechanism \RAVk{k} is
strategyproof in expectation.  In contrast, for every $k\ge4$, there exists a profile under which some agent can strictly lower her expected distance by misreporting.
\end{mainthm}

The cases $k=1$ and $k=2$ are random dictatorship and the Proportional
Mechanism of Lu et al.~\cite{LuEtAl2010}; the rule is new from $k=3$ onward.

For $k\geq 4$, we show that \RAVk{k} is not strategyproof (in expectation).

Our second theorem gives the social cost approximation guarantee of the RAV family of mechanisms. 

\begin{mainthm}[Social-cost approximation]\label{thm:intro-approximation}
For every \(k \geq 2\) and every profile \(\mathbf{x}\),
\[
  \E\!\left[
    \SC\bigl(\mathbf{x},\RAVk{k}(\mathbf{x})\bigr)
  \right]
  \leq 4(k-1)\OPT_k(\mathbf{x}).
\]
Thus, \(\RAVk{k}\) achieves a \(4(k-1)\)-approximation for social cost.
In particular, the strategyproof-in-expectation mechanism \(\RAVk{3}\)
achieves an approximation factor of \(8\).
\end{mainthm}

\subsection{Related work}

Our work belongs to approximate mechanism design without money
\cite{ProcacciaTennenholtz2013}, within the broader programs of algorithmic
mechanism design \cite{NisanRonen2001,NisanEtAl2007} and computational social
choice \cite{BrandtEtAl2016}.  For one facility, the single-peaked tradition
goes back to Black \cite{Black1948,Black1958}, while Moulin's generalized
median characterization is foundational for strategyproof social choice on a
line \cite{Moulin1980,Moulin1988}.  Later work studied randomized mechanisms
on networks \cite{AlonEtAl2010} and approximation--variance tradeoffs
\cite{ProcacciaWajcZhang2018,LuEtAl2010}.

For two facilities, the Proportional Mechanism of Lu et
al. is strategyproof in expectation and has approximation factor $4$ in every
metric space \cite{LuEtAl2010}; earlier work established lower bounds and
line-specific guarantees \cite{LuWangZhou2009}.  Lu et al. showed that their
sequential extension to three facilities is manipulable, and gave a different
truthful mechanism with approximation linear in $n$.

For $k\ge3$, Escoffier et al. studied the many-facility regime in which the
number of facilities is one less than the number of agents
\cite{EscoffierEtAl2011}.  Fotakis and Tzamos proved that no deterministic
\emph{anonymous} strategyproof mechanism has a bounded approximation ratio,
even with $k+1$ agents \cite{FotakisTzamos2014}.  Their Equal Cost mechanism works for all $k$
and $n$, may place facilities away from reports, is group-strategyproof for
concave distance costs, and has social-cost ratio at most $n$
\cite{FotakisTzamosConcave}.  In the special case $n=k+1$, their Pick the Loser
mechanism achieves factor $2$.  These results illustrate the force of the
unrestricted model considered here: we impose neither the relation $n=k+1$
nor a change in the agents' cost rule.

Winner-imposing mechanisms obtain factor $4k$ for every $k$
\cite{FotakisTzamosWinnerImposing}.  In that model, an agent selected to host a
facility must use that facility rather than the nearest open one, so the cost
model differs from ours.  Other positive results impose structural
restrictions, including perturbation stability and the case $n=k+1$
\cite{FotakisPatsilinakos2022,FotakisTzamosConcave}.

A complementary computational literature optimizes expected social cost under
an assumed distribution of locations.  Golowich, Narasimhan, and Parkes study
multi-facility generalized-median rules and data-driven almost-strategyproof
mechanisms with low ex-post regret in that setting \cite{GolowichEtAl2018}.  Our results
instead give exact incentive compatibility and worst-case guarantees for every
profile.  Nearby work also considers additive error
\cite{GolombTzamos2017}, capacities \cite{AuricchioWangZhang2025}, Bayesian
percentile mechanisms based on optimal transport \cite{AuricchioZhang2025},
and private acceptable-facility sets \cite{ShaEtAl2025}.

The adjacent-gap probability rule is closely related to the one-dimensional specialization of
fixed-size volume sampling, introduced for low-dimensional approximation by
Deshpande et al. and subsequently developed algorithmically
\cite{DeshpandeEtAl2006,DeshpandeRademacher2010}.  The identification of the
gap product with the Gram determinant of the segments from the anchor to the
selected reports makes this connection exact and explains the word ``Volume'' in the mechanism's
name.  Our welfare proof nevertheless works directly with gaps and does not
use linear-algebraic sampling machinery.  The new ingredients here are the
facility-location mechanism, the three-facility truthfulness proof, the direct
adjacent-gap welfare comparison, and the exact incentive failure from four
facilities onward.

For the many further variants of strategyproof facility location, we refer
to the survey of Chan et al.~\cite{ChanEtAl2021}.

\paragraph{Road map.} 

\Cref{sec:model} introduces the model, and \cref{sec:mechanism} defines
\RAVk{3} and its generalization to arbitrary k.
\Cref{sec:strategyproofness} proves that \RAVk{3} is strategyproof in expectation. 
\cref{sec:approximation} proves \cref{thm:intro-approximation}. Finally, \cref{sec:general-incentive-failure} constructs, for every $k\geq 4$, an instance in which $\RAV_k$ is manipulable, thereby completing the proof of \cref{thm:intro-boundary}.

\section{Preliminaries}\label{sec:model}

$N=\set{1,\ldots,n}$ be the set of agents. Each agent $i$ has a location 
$x_i\in\R$. We denote $\mathbf{x}=(x_1,...,x_n)$ as the location profile. 

A nonempty set $F\subseteq\R$ with $1\le\abs{F}\le k$ specifies the opened
facilities.    For a location $x_i\in\R$ and a nonempty facility set $F$, define the distance from
$x_i$ to its nearest open facility by
\[
  d(x_i,F):=\min_{f\in F}\abs{x_i-f}.
\]
Given a set $F$ of opened facilities, agent
$i$'s cost is
\[
  \cost(x_i,F):=d(x_i,F),
\]
and the \emph{social cost} is
\[
\SC(\mathbf{x},F):=\sum_{i=1}^{n}\cost(x_i,F)
           =\sum_{i=1}^{n}d(x_i,F).
\]
For a location profile $\mathbf{x}=(x_1,\ldots,x_n)$, the optimal $k$-facility social cost is
\[
\OPT_k(\mathbf{x}):=\min_{F\subseteq\R,\,1\le\abs{F}\le k}\SC(\mathbf{x},F).
\]

A \emph{randomized mechanism} maps the reports to a probability distribution
over facility sets.  For a fixed report profile $\mathbf{x}$, the notation
$M(\mathbf{x})$ denotes the resulting set-valued random variable.  Consequently,
$d(x_i,M(\mathbf{x}))$ is the random cost of a location $x_i$, while
$\SC(\mathbf{x},M(\mathbf{x}))$ is the random social cost at the  profile
$\mathbf{x}$.  Expectations of these quantities are taken over the mechanism's own random draws.

\begin{definition}[Strategyproofness in expectation]\label{def:sp}
A randomized mechanism $M$ is \emph{strategyproof in expectation} if, for
every agent $i$, every true location $t$, every alternative report $r$, and
every fixed list $\mathbf{x}_{-i}$ of the other reports,
\[
  \E\!\left[d\bigl(t,M(t,\mathbf{x}_{-i})\bigr)\right]
  \le
  \E\!\left[d\bigl(t,M(r,\mathbf{x}_{-i})\bigr)\right].
\]
\end{definition}

This property is also called \emph{truthfulness in expectation}.  It rules out
a profitable unilateral misreport before the random draw.  

A randomized $k$-facility mechanism is an $\alpha$-approximation for social cost if, for every profile $\mathbf{x}$,
\[
\E\!\left[\SC(\mathbf{x},M(\mathbf{x}))\right]\le \alpha\OPT_k(\mathbf{x}).
\]
Our main objective is to design a strategyproof-in-expectation mechanism that achieves a constant-factor approximation to the optimal social cost.

Throughout the paper, we assume without loss of generality that $k\leq n$. Otherwise, we can simply open a facility at each agent’s reported location, which is both optimal and strategyproof.

\section{The Random-Anchor Volume mechanism}\label{sec:mechanism}

We first describe the Random-Anchor Volume (\RAV) mechanism for $k=3$. 

\medskip

The mechanism first selects an agent uniformly at random and designates that agent’s reported location as the \emph{anchor}. It then randomly selects two additional reports, with each pair’s selection probability determined by how well the three resulting locations are separated along the line.

\bigskip

\noindent \textbf{The anchor-relative adjacent-gap weight}

\bigskip

Fix an anchor location \(p\) and two candidate
locations \(x,y \in \mathbb{R} \). Arrange the three locations in nondecreasing order,
\[
    z_0\le z_1\le z_2,
\]
and define their \emph{adjacent-gap weight} by
\begin{equation}\label{eq:delta-definition}
    \Delta_p(x,y):=(z_1-z_0)(z_2-z_1).
\end{equation}
Thus, the weight is the product of the two consecutive gaps among
\(p,x,y\). In particular, it is positive exactly when the three locations are
distinct, and it is larger when both gaps are large.

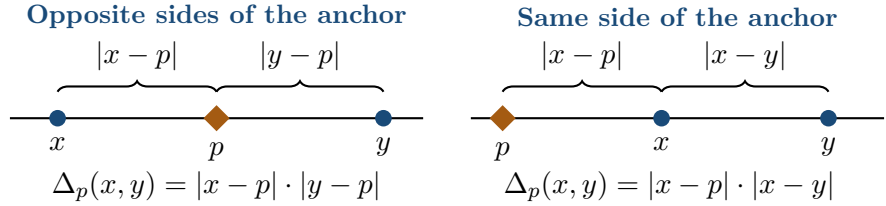
\begin{figure}[htbp]
\centering
\begin{tikzpicture}[x=1.05cm,y=1cm]
  \node[font=\small\bfseries\color{ravblue}] at (2.6,2.75) {Opposite sides of the anchor};
  \draw[thick] (0,1.4)--(5.2,1.4);
  \node[ravpoint,label=below:$x$] at (0.6,1.4) {};
  \node[ravanchor,label=below:$p$] at (2.6,1.4) {};
  \node[ravpoint,label=below:$y$] at (4.7,1.4) {};
  \draw[gapbrace] (0.6,1.75)--node[above=4pt] {$\abs{x-p}$}(2.6,1.75);
  \draw[gapbrace] (2.6,1.75)--node[above=4pt] {$\abs{y-p}$}(4.7,1.75);
  \node at (2.6,0.5) {$\Delta_p(x,y)=\abs{x-p}\cdot \abs{y-p}$};

  \node[font=\small\bfseries\color{ravblue}] at (8.4,2.75) {Same side of the anchor};
  \draw[thick] (5.8,1.4)--(11,1.4);
  \node[ravanchor,label=below:$p$] at (6.2,1.4) {};
  \node[ravpoint,label=below:$x$] at (8.2,1.4) {};
  \node[ravpoint,label=below:$y$] at (10.3,1.4) {};
  \draw[gapbrace] (6.2,1.75)--node[above=4pt] {$\abs{x-p}$}(8.2,1.75);
  \draw[gapbrace] (8.2,1.75)--node[above=4pt] {$\abs{x-y}$}(10.3,1.75);
  \node at (8.3,0.5) {$\Delta_p(x,y)=\abs{x-p}\cdot\abs{x-y}$};
\end{tikzpicture}
\end{figure}
\medskip 

\noindent \textbf{Description of $\RAV_3$.}
Given a reported profile \(\mathbf{x}\), the mechanism proceeds in two stages.

\begin{enumerate}
\item Choose an anchor agent \(I\) uniformly at random from \(N\), and open
a facility at the anchor's reported location \(x_I\).

\item Suppose that \(I=a\), and write \(p=x_a\). For each pair of
non-anchor agents \(\{j,j'\}\subseteq N\setminus\{a\}\), assign the weight
\[
    \Delta_p(x_j,x_{j'}).
\]
Let
\[
    D_a(\mathbf{x})
    :=
    \sum_{\{j,j'\}\in\binom{N\setminus\{a\}}{2}}
    \Delta_p(x_j,x_{j'})
\]
be the total weight of all such pairs.
\begin{itemize}
\item If \(D_a({\mathbf{x}})>0\), select a pair
\(\{j,j'\}\in\binom{N\setminus\{a\}}{2} \) with probability
\[
    \Pr\!\left( \mathcal{S}=\{j,j'\}\mid I=a\right)
    =
    \frac{\Delta_p({x}_j,{x}_{j'})}
         {D_a({\mathbf{x}})},
\]
and open the remaining two facilities at
\(x_j\) and \(x_{j'}\). 
\item If \(D_a(\mathbf{x})=0\), open a facility at every distinct reported
location. We call this the \emph{fallback rule}.
\end{itemize}
\end{enumerate}

In summary, \(\RAVk{3}\) chooses one report uniformly as an anchor and then
chooses the other two reports with probability proportional to the product of
the two gaps created by the resulting three locations.

Note that the fallback rule is invoked exactly when the reported profile has at most two
distinct coordinates, so opening at every distinct report uses at most two
facilities.

\begin{example}[A complete four-agent lottery]\label{ex:mechanism}
Consider four agents with reported profile
\[
\mathbf{x}=(0,1,4,7).
\]
Suppose first that agent \(1\), whose report is \(x_1=0\), is selected as the
anchor; thus, \(p=x_1=0\). The three candidate pairs of non-anchor agents are
\(\set{2,3}\), \(\set{2,4}\), and \(\set{3,4}\). Their adjacent-gap weights
are
\[
\Delta_p(x_2,x_3)=3,\qquad
\Delta_p(x_2,x_4)=6,\qquad
\Delta_p(x_3,x_4)=12.
\]
Since the total weight is \(21\), the corresponding conditional probabilities
are \(1/7\), \(2/7\), and \(4/7\), respectively.

Each of the four agents is selected as the anchor with probability \(1/4\).
Repeating the calculation for each possible anchor and combining the probabilities whenever different anchor choices produce the same facility set gives the following complete output lottery. Let \(\mathbf{F}\) denote the random facility set. Then
\[
\begin{aligned}
 \Pr(\mathbf{F}=\set{0,1,4})&=\frac{73}{672},
 &\Pr(\mathbf{F}=\set{0,1,7})&=\frac{53}{252},\\
 \Pr(\mathbf{F}=\set{0,4,7})&=\frac{191}{504},
 &\Pr(\mathbf{F}=\set{1,4,7})&=\frac{29}{96}.
\end{aligned}
\]
\end{example}

\subsection{\RAV mechanism for general \texorpdfstring{\(k\)}{k}}

We now describe the Random-Anchor Volume mechanism for an arbitrary number
\(k\) of facilities.
\medskip

The mechanism first selects an agent uniformly at random and designates that
agent's reported location as the \emph{anchor}. It then randomly selects
\(k-1\) additional reports, with each candidate set's selection probability
determined by how well the resulting \(k\) locations are separated along the
line.

\bigskip

\noindent\textbf{The anchor-relative adjacent-gap volume}

\bigskip

Fix an anchor location \(p\) and other \(k-1\) candidate locations \(\hat{x}_1,\ldots,\hat{x}_{k-1}\in\mathbb{R}\).
Arrange these \(k\) locations in nondecreasing order,
\[
    z_0\leq z_1\leq\cdots\leq z_{k-1},
\]
and define their \emph{adjacent-gap volume} by
\begin{equation}\label{eq:delta-k-definition}
    \Delta_p(\hat{x}_1,\ldots,\hat{x}_{k-1})
    :=
    \prod_{\ell=1}^{k-1}(z_\ell-z_{\ell-1}).
\end{equation}
Thus, the volume is the product of the \(k-1\) consecutive gaps among
\(p,\hat{x}_1,\ldots,\hat{x}_{k-1}\). In particular, it is positive exactly when the
\(k\) locations are pairwise distinct. It assigns greater weight to
configurations in which all consecutive gaps are large.

\medskip

For a set of agents
\(
    S=\{j_1,\ldots,j_{k-1}\},
\)
we use the abbreviation
\[
    \Delta_p(\mathbf{x}_S)
    :=
    \Delta_p(x_{j_1},\ldots,x_{j_{k-1}}).
\]
Note that the definition is independent of the ordering of the agents in \(S\).

\bigskip

\noindent\textbf{Description of \(\RAV_k\).}

\medskip

Given a reported profile \(\mathbf{x}\), the mechanism proceeds in two
stages.

\begin{enumerate}
\item Choose an anchor agent \(I\) uniformly at random from \(N\), and open
a facility at the anchor's reported location \(x_I\).

\item Suppose that \(I=a\), and write \(p=x_a\). For each set of
\(k-1\) non-anchor agents
\(S\in\binom{N\setminus\{a\}}{k-1}\),
assign the volume
\(
    \Delta_p(\mathbf{x}_S).
\)
Let
\[
    D_a^{(k)}(\mathbf{x})
    :=
    \sum_{S\in\binom{N\setminus\{a\}}{k-1}}
    \Delta_p(\mathbf{x}_S)
\]
be the total volume of all such candidate sets.

\begin{itemize}
\item If \(D_a^{(k)}(\mathbf{x})>0\), select a set
\(S\in\binom{N\setminus\{a\}}{k-1}\) with probability
\[
    \Pr\!\left(\mathcal{S}=S\mid I=a\right)
    =
    \frac{\Delta_p(\mathbf{x}_S)}
         {D_a^{(k)}(\mathbf{x})},
\]
and open the remaining \(k-1\) facilities at the reported locations
\[
    \{x_j:j\in S\}.
\]

\item If \(D_a^{(k)}(\mathbf{x})=0\), open a facility at every distinct
reported location. We call this the \emph{fallback rule}.
\end{itemize}
\end{enumerate}

Observe that \(D_a^{(k)}(\mathbf{x})=0\) exactly when the profile contains
fewer than \(k\) distinct reported locations. Thus, whenever the fallback
rule is invoked, its outcome is independent of the identity of the anchor.
If the mechanism is required to specify exactly \(k\) facility locations,
any remaining facilities may be colocated with facilities that have already
been opened.

\section{Strategyproofness of \RAVk{3}}\label{sec:strategyproofness}

\begin{theorem}\label{thm:rav3-sp}
The mechanism \(\RAVk{3}\) is strategyproof in expectation.
\end{theorem}

We first observe that, because the anchor is chosen uniformly and independently
of the reported locations, it suffices to establish the incentive inequality
separately for each possible anchor identity.

\begin{lemma}[Anchorwise reduction]
\label{lem:anchorwise-reduction}
Fix an agent \(i\), its true location \(t\), the reports of the other agents
\(\mathbf{x}_{-i}\), and an arbitrary alternative report \(r\). Suppose that,
for every possible anchor agent \(a\in N\),
\begin{equation}\label{eq:AnchorCond}
  \E\!\left[
    d\bigl(t,\RAVk{3}(t,\mathbf{x}_{-i})\bigr)
    \given I=a
  \right]
  \leq
  \E\!\left[
    d\bigl(t,\RAVk{3}(r,\mathbf{x}_{-i})\bigr)
    \given I=a
  \right],
\end{equation}
where the expectation is taken over the second-stage randomization, including
the fallback rule. Then
\[
  \E\!\left[
    d\bigl(t,\RAVk{3}(t,\mathbf{x}_{-i})\bigr)
  \right]
  \leq
  \E\!\left[
    d\bigl(t,\RAVk{3}(r,\mathbf{x}_{-i})\bigr)
  \right].
\]
Consequently, if~\eqref{eq:AnchorCond} holds for every choice of
\(i\), \(t\), \(\mathbf{x}_{-i}\), and \(r\), then \(\RAVk{3}\) is
strategyproof in expectation.
\end{lemma}

\begin{proof}
The first stage selects the anchor agent uniformly and independently of the
reported locations. Hence, under both the truthful report \(t\) and the
alternative report \(r\), we have
\[
  \Pr(I=a)=\frac{1}{n}
  \qquad\text{for every }a\in N.
\]
By the law of total expectation,
\[
\begin{aligned}
  \E\!\left[
    d\bigl(t,\RAVk{3}(t,\mathbf{x}_{-i})\bigr)
  \right]
  &=
  \frac{1}{n}\sum_{a\in N}
  \E\!\left[
    d\bigl(t,\RAVk{3}(t,\mathbf{x}_{-i})\bigr)
    \given I=a
  \right] \\
  &\leq
  \frac{1}{n}\sum_{a\in N}
  \E\!\left[
    d\bigl(t,\RAVk{3}(r,\mathbf{x}_{-i})\bigr)
    \given I=a
  \right] \\
  &=
  \E\!\left[
    d\bigl(t,\RAVk{3}(r,\mathbf{x}_{-i})\bigr)
  \right],
\end{aligned}
\]
where the inequality follows from~\eqref{eq:AnchorCond}.
\end{proof}
In words, Lemma~\ref{lem:anchorwise-reduction} shows that it suffices to prove
that \(\RAVk{3}\) is truthful in expectation conditional on each possible
anchor identity.

\subsection{Truthfulness in expectation for a fixed anchor}
\label{sec:fixed-anchor}

In this section, we verify that \(\RAV_3\) is truthful in expectation
conditional on each possible anchor identity. Specifically, we prove that
\Cref{eq:AnchorCond} holds for every choice of \(i\), \(t\), \(\mathbf{x}_{-i}\), and \(r\). By \Cref{lem:anchorwise-reduction}, this establishes
\Cref{thm:rav3-sp}.

Fix an arbitrary agent with true location \(t\) and an arbitrary
alternative report \(r\). If this agent is selected as the anchor, then under
truthful reporting a facility is opened at \(t\), so their cost is zero.
Hence, \Cref{eq:AnchorCond} holds immediately in this case.

It therefore remains to consider the case in which some other agent is selected
as the anchor. Let \(p\) denote the anchor's reported location. If \(p=t\), then
the deviating agent is already served by the facility at the anchor, regardless
of whether it reports \(t\) or \(r\). Thus, \Cref{eq:AnchorCond} again holds
trivially.

Suppose, therefore, that \(p\neq t\). By translating and, if necessary,
reflecting the line, we may assume without loss of generality that \(p=0\) and
\(t>0\). We continue to use the same notation for the transformed locations.
Let
\(
    (y_1,\ldots,y_{n-2})
\)
denote the reports of all agents other than the anchor and the deviating agent.

For any coordinate \(x\in \mathbb{R}\), denote
\(
 c(x):=d(t,\set{0,x})=\min\{t,\abs{t-x}\}.
\)
Thus \(c(x)\) is the deviator's cost when the anchor and \(x\) are open.
All indices in the sums below range over \(\{1,\ldots,n-2\}\), i.e., the reports of all agents other than the anchor and the deviating agents.

We first express the deviating agent's expected cost under truthful report \(t\)
and under an alternative report \(r\). Define
\begin{align*}
    Z
    &:= \sum_{i<j}\Delta_0(y_i,y_j),&
    A
    &:= \sum_{i<j}\Delta_0(y_i,y_j)
        \min\{c(y_i),c(y_j)\},\\
    U(z)
    &:= \sum_i\Delta_0(z,y_i),&
    B(r)
    &:= \sum_i\Delta_0(r,y_i)
        \min\{c(r),c(y_i)\}.
\end{align*}
Here, \(Z\) is the total weight of the candidate pairs that do not contain the
deviating agent, whereas \(U(z)\) is the total weight of the candidate pairs
that contain the deviating agent when it reports \(z\).

Whenever the corresponding total weight is positive, the conditional expected
costs under the truthful report \(t\) and the alternative report \(r\) are,
respectively,
\[
    C_t=\frac{A}{Z+U(t)}
    \qquad\text{and}\qquad
    C_r=\frac{A+B(r)}{Z+U(r)}.
\]

If \(Z+U(t)=0\), then the fallback rule opens a facility at every distinct
reported location, including \(t\). Thus, the deviating agent incurs zero cost
under truthful reporting, and \Cref{eq:AnchorCond} holds immediately.

Now suppose that \(Z+U(r)=0\). Then the profile induced by the alternative
report \(r\) contains at most two distinct reported locations. Replacing \(r\)
by the truthful report \(t\) therefore produces a profile with at most three
distinct reported locations. If there are fewer than three distinct locations,
the fallback rule opens a facility at \(t\). If there are exactly three, every
pair of non-anchor reports with positive weight contains one report at each of
the two locations distinct from the anchor, one of which is \(t\). Hence, every
possible outcome under truthful reporting again opens a facility at \(t\).
The deviating agent therefore incurs zero cost under truthful reporting, so
\Cref{eq:AnchorCond} holds in this case as well. Thus, we may assume that both $Z+U(t)>0$ and $Z+U(r)>0$. 

The difference between the two conditional expected costs, namely, the cost
under the alternative report minus the cost under truthful reporting, is
\begin{align*}
    C_r-C_t
    &= \frac{A+B(r)}{Z+U(r)}
       -\frac{A}{Z+U(t)} \\
    &= \frac{
        B(r)\bigl(Z+U(t)\bigr)
        -A\bigl(U(r)-U(t)\bigr)}
       {(Z+U(r))(Z+U(t))}.
\end{align*}
By the preceding discussion, the denominator is positive. Therefore,
\(C_r\geq C_t\) if and only if the numerator is nonnegative, that is,
\begin{equation}\label{eq:fixed-anchor-numerator}
    B(r)\bigl(Z+U(t)\bigr)
    -A\bigl(U(r)-U(t)\bigr)
    \geq 0.
\end{equation}
Hence, to establish \Cref{eq:AnchorCond}, it suffices to prove
\Cref{eq:fixed-anchor-numerator}.

\subsubsection{Grouping the numerator into triples}

Remove for the moment the visibly nonnegative term \(B(r)U(t)\).  Expanding
the rest of the left-hand side of \eqref{eq:fixed-anchor-numerator} gives
\begin{align}
 &B(r)Z-A\bigl(U(r)-U(t)\bigr)\notag\\
 &\quad=\sum_i\sum_{j<\ell}\Delta_0(y_j,y_\ell)
 \Bigl[\Delta_0(r,y_i)\min\{c(r),c(y_i)\}
 -\bigl(\Delta_0(r,y_i)-\Delta_0(t,y_i)\bigr)
   \min\{c(y_j),c(y_\ell)\}\Bigr].
 \label{eq:ungrouped-incentive}
\end{align}
Here \(i\) may equal \(j\) or \(\ell\); only the two indices in the unchanged
pair must be distinct.  Each summand is therefore indexed by a distinguished
identity \(i\) and an unchanged pair \(\{j,\ell\}\).  When all three identities
are distinct, their three choices of distinguished identity form one
three-term cyclic block.  Accordingly
define
\begin{equation}\label{eq:cyclic-contribution}
\begin{aligned}
 \Phi_r(x,y,z):={}&\Delta_0(y,z)\Bigl[\Delta_0(r,x)\min\{c(r),c(x)\}
 -\bigl(\Delta_0(r,x)-\Delta_0(t,x)\bigr)\min\{c(y),c(z)\}\Bigr]\\
 &+\Delta_0(z,x)\Bigl[\Delta_0(r,y)\min\{c(r),c(y)\}
 -\bigl(\Delta_0(r,y)-\Delta_0(t,y)\bigr)\min\{c(z),c(x)\}\Bigr]\\
 &+\Delta_0(x,y)\Bigl[\Delta_0(r,z)\min\{c(r),c(z)\}
 -\bigl(\Delta_0(r,z)-\Delta_0(t,z)\bigr)\min\{c(x),c(y)\}\Bigr].
\end{aligned}
\end{equation}
Each of the three terms chooses one coordinate as the distinguished report;
the other two form the unchanged pair.  Since \(\Delta_0\) and the minimum are
symmetric in that pair, \(\Phi_r\) is symmetric in \(x,y,z\).

If the distinguished identity belongs to the unchanged pair, only two
identities are involved.  In \(\Phi_r(y_i,y_i,y_j)\), and similarly in
\(\Phi_r(y_i,y_j,y_j)\), the relevant summand appears twice while the third
summand vanishes because \(\Delta_0(y_i,y_i)=0\).  This explains the factor
\(1/2\) below.  Thus \eqref{eq:ungrouped-incentive} gives the exact
decomposition
\begin{align}
 &B(r)\bigl(Z+U(t)\bigr)-A\bigl(U(r)-U(t)\bigr)\notag\\
 &\quad=B(r)U(t)+\sum_{i<j<\ell}\Phi_r(y_i,y_j,y_\ell)
 +\frac12\sum_{i<j}
 \bigl[\Phi_r(y_i,y_i,y_j)+\Phi_r(y_i,y_j,y_j)\bigr].
 \label{eq:orbit-decomposition}
\end{align}

In each cyclic block, the distinguished report is paired with the deviator,
so \(\Delta_0(t,x)\) changes to \(\Delta_0(r,x)\); the other two reports form
the unchanged pair.  A negative contribution from one distinguished role is
compensated by the other two roles---hence the name.

\begin{lemma}[Cyclic compensation]\label{lem:cyclic-compensation}
Fix the \(t>0\) above and \(x,y,z\in\mathbb R\), with repeated coordinates
allowed.  Then
\[
 \Phi_r(x,y,z)\ge0
 \qquad\text{for every }r\in\mathbb R.
\]
\end{lemma}

Because \(B(r),U(t)\ge0\), the lemma and \eqref{eq:orbit-decomposition}
imply \Cref{eq:fixed-anchor-numerator}.  We now prove the lemma.

\subsubsection{Three properties of anchor-relative weight}

The compensation argument uses the following three geometric facts about
\(\Delta_0\).

\begin{lemma}[Geometry of \(\Delta_0\)]\label{lem:delta0-geometry}
\begin{enumerate}[label=\textup{(\roman*)}]
\item \emph{Piecewise shape.} For fixed \(y>0\),
\[
 \Delta_0(x,y)=
 \begin{cases}
  (-x)y,&x\le0,\\
  x(y-x),&0\le x\le y,\\
  y(x-y),&x\ge y.
 \end{cases}
\]
For \(y<0\), reflect the picture along the $y$-axis; for \(y=0\), the function is zero.
Consequently, \(x\mapsto\Delta_0(x,y)\) is concave on every interval
containing neither \(0\) nor \(y\) in its interior.

\item \emph{Four-point exchange.} For arbitrary \(r,x_1,x_2,x_3\in\R\),
\[
 \Delta_0(r,x_1)\Delta_0(x_2,x_3)
 \le
 \Delta_0(r,x_2)\Delta_0(x_1,x_3)
 +\Delta_0(r,x_3)\Delta_0(x_1,x_2).
\]

\item \emph{One-centre product bound.} If \(x>0\), \(0\le\rho\le x\), and
\(y,z\notin(x-\rho,x+\rho)\), then
\[
 2\Delta_0(x,y)\Delta_0(x,z)
 \ge x\rho\,\Delta_0(y,z).
\]
\end{enumerate}
\end{lemma}

\begin{proof}
Part~(i) follows by sorting \(0,x,y\) in the definition of \(\Delta_0\).
Between its zeros \(0\) and \(y\), the graph is a concave quadratic arch;
outside that segment it is linear.

For the other two parts, represent a nonzero coordinate by the rooted segment
from \(0\) to that coordinate, with the two rays treated as separate arms.
For nonzero \(x,y\), define the normalized separation
\[
 \eta(x,y):=\frac{\Delta_0(x,y)}{\abs{x}\abs{y}}
 =
 \begin{cases}
  1,
    &x,y\text{ lie on opposite rays},\\[2mm]
  1-\dfrac{\min\{\abs{x},\abs{y}\}}
           {\max\{\abs{x},\abs{y}\}},
    &x,y\text{ lie on the same ray}.
 \end{cases}
\]
This is the fraction of the longer segment not shared by the shorter one.
It is a metric on \(\R\setminus\set{0}\).  Indeed, the only nontrivial
triangle has all three points on one ray.  If
\(0<\abs{x}\le\abs{y}\le\abs{z}\), then
\[
 1-\eta(x,z)
 =\bigl(1-\eta(x,y)\bigr)
  \bigl(1-\eta(y,z)\bigr),
\]
which implies \(\eta(x,z)\le\eta(x,y)+\eta(y,z)\).  A triangle involving
both rays has two sides of length \(1\) and a third of length at most \(1\),
so its inequalities are immediate.

For part~(ii), if one of the four points is \(0\), all three pairing products
vanish.  Otherwise divide by the product of their four radii.  The desired
claim becomes the corresponding inequality for products of \(\eta\)-values.
If two points lie on each ray, the two crossed pairing products equal \(1\),
whereas the same-ray product is at most \(1\).  If three points lie on one ray
and the fourth on the other, the three products reduce to the three normalized
distances among the same-ray points, and the claim follows from the triangle
inequality.

It remains to put all four points on one ray.  After reflection, write their
ordered radii as
\[
 0<u_1\le u_2\le u_3\le u_4.
\]
The pairing \((u_1,u_3),(u_2,u_4)\) dominates
\((u_1,u_2),(u_3,u_4)\) termwise.  It also dominates the third pairing,
because
\[
 \eta(u_1,u_3)\eta(u_2,u_4)
 -\eta(u_1,u_4)\eta(u_2,u_3)
 =\frac{(u_2-u_1)(u_4-u_3)}{u_3u_4}\ge0.
\]
Finally,
\[
\begin{aligned}
 &\eta(u_1,u_2)\eta(u_3,u_4)
 +\eta(u_1,u_4)\eta(u_2,u_3)\\
 &\qquad-\eta(u_1,u_3)\eta(u_2,u_4)
 =\frac{(u_2-u_1)(u_3-u_2)(u_4-u_3)}
        {u_2u_3u_4}\ge0.
\end{aligned}
\]
Thus the largest pairing product is at most the other two combined.
Multiplying back by the four radii proves part~(ii).

For part~(iii), the cases \(\rho=0\), \(y=0\), or \(z=0\) are immediate.
Otherwise put
\[
 \lambda:=\frac{\rho}{x},\qquad
 u:=\eta(x,y),\qquad v:=\eta(x,z),
\]
and assume \(y\le z\).  Dividing the desired inequality by
\(x^2\abs{y}\abs{z}\) reduces it to
\[
 2uv\ge\lambda\eta(y,z).
\]
If both \(y,z\) lie to the left of the excluded interval, then
\(u,v\ge\lambda\).  The triangle inequality gives
\(\eta(y,z)\le u+v\), and
\[
 2uv-\lambda(u+v)
 =u(v-\lambda)+v(u-\lambda)\ge0.
\]
If \(y\le x-\rho\) and \(z\ge x+\rho\), then
\[
 u\ge\lambda,\qquad v\ge\frac{\lambda}{1+\lambda},
 \qquad \eta(y,z)=u+v-uv.
\]
The expression \(2uv-\lambda(u+v-uv)\) is increasing in \(u\), because its
derivative is
\((2+\lambda)v-\lambda\ge\lambda/(1+\lambda)\).  At \(u=\lambda\)
it equals
\[
 \lambda\bigl((1+\lambda)v-\lambda\bigr)\ge0.
\]
Finally, if both points lie to the right, then
\[
 u\ge\frac{\lambda}{1+\lambda}\ge\frac{\lambda}{2},
 \qquad \eta(y,z)\le v,
\]
so \(2uv\ge\lambda v\ge\lambda\eta(y,z)\).  This proves part~(iii).
\end{proof}

\subsubsection{Proof of \cref{lem:cyclic-compensation}}

The definition of \(\Phi_r\) is symmetric, so relabel its three arguments as
\(y_1,y_2,y_3\) such that
\[
 \alpha:=c(y_1)\le\beta:=c(y_2)\le\gamma:=c(y_3).
\]
The unchanged pair \(\{y_2,y_3\}\) serves the deviator at cost \(\beta\),
the other two unchanged pairs serve her at cost \(\alpha\), and a pair
containing \(r,y_i\) serves her at cost \(\min\{c(r),c(y_i)\}\).  Collecting
these three roles in \eqref{eq:cyclic-contribution} gives
\begin{equation}\label{eq:triple-collected}
\begin{aligned}
 \Phi_r(y_1,y_2,y_3)={}
 &\Delta_0(y_2,y_3)\Bigl[
       \beta\Delta_0(t,y_1)
       +\Delta_0(r,y_1)\bigl(\min\{c(r),\alpha\}-\beta\bigr)
     \Bigr]\\
 &+\Delta_0(y_1,y_3)\Bigl[
       \alpha\Delta_0(t,y_2)
       +\Delta_0(r,y_2)\bigl(\min\{c(r),\beta\}-\alpha\bigr)
     \Bigr]\\
 &+\Delta_0(y_1,y_2)\Bigl[
       \alpha\Delta_0(t,y_3)
       +\Delta_0(r,y_3)\bigl(\min\{c(r),\gamma\}-\alpha\bigr)
     \Bigr].
\end{aligned}
\end{equation}

There are two regimes.  If \(c(r)\ge\beta\), the misreport serves the agent
no better than the second-best unchanged report, and the four-point exchange
gives compensation directly.  If \(c(r)<\beta\), then
\(r=t\pm c(r)\); along either direction, the cyclic contribution is a
one-variable piecewise cubic.  We show that this function has no negative
interior minimum.

\paragraph{Case 1: \(c(r)\ge\beta\).}
In this range, \eqref{eq:triple-collected} becomes
\begin{align*}
 \Phi_r(y_1,y_2,y_3)={}
 &\beta\Delta_0(y_2,y_3)\Delta_0(t,y_1)
 +\alpha\Delta_0(y_1,y_3)\Delta_0(t,y_2)\\
 &+\alpha\Delta_0(y_1,y_2)\Delta_0(t,y_3)\\
 &+(\beta-\alpha)\Bigl[
   \Delta_0(y_1,y_3)\Delta_0(r,y_2)
  +\Delta_0(y_1,y_2)\Delta_0(r,y_3)
  -\Delta_0(y_2,y_3)\Delta_0(r,y_1)\Bigr]\\
 &+\Delta_0(y_1,y_2)\Delta_0(r,y_3)
   \bigl(\min\{c(r),\gamma\}-\beta\bigr).
\end{align*}
Every term outside the bracket is nonnegative.  Part~(ii) of
\cref{lem:delta0-geometry} gives
\[
 \Delta_0(y_2,y_3)\Delta_0(r,y_1)
 \le \Delta_0(y_1,y_3)\Delta_0(r,y_2)
   +\Delta_0(y_1,y_2)\Delta_0(r,y_3),
\]
so the bracket is nonnegative as well.  Hence
\(\Phi_r(y_1,y_2,y_3)\ge0\) whenever \(c(r)\ge\beta\).

\paragraph{Case 2: \(c(r)<\beta\).}
Now \(\beta>0\), and \(c(r)<\beta\le t\) implies \(c(r)<t\).  Hence
\(c(r)=\abs{t-r}\), so
\[
 r=t+\sigma c(r),
 \qquad \sigma\in\{-1,+1\}.
\]
Fix \(\sigma\).  Using a dummy variable \(s\), define
\[
 \begin{aligned}
 f(s)&:=\Phi_{t+\sigma s}(y_1,y_2,y_3)
       &&(0\le s\le\beta),\\
 h_i(s)&:=\Delta_0(t+\sigma s,y_i)
       &&(-\alpha\le s\le\beta).
 \end{aligned}
\]
In the definition of \(f(s)\), the selected-report cost is
\(c(t+\sigma s)=s\).  The larger domain of \(h_i\) is intentional: the
endpoint estimate below uses \(h_i(-\alpha)\).  We have \(f(0)=0\), while
\(f(\beta)\ge0\) by Case~1.  It therefore suffices to show that
\begin{enumerate}[label=\textup{(\alph*)}]
\item \(f\) is nondecreasing on \([0,\alpha]\); and
\item the minimum of \(f\) on \([\alpha,\beta]\) is attained at an endpoint.
\end{enumerate}

We first record the branch geometry.  If \(y_i=0\), then \(h_i\equiv0\).
Otherwise its branch points occur when \(t+\sigma s\in\{0,y_i\}\), at
parameters whose absolute values are \(t\) and \(\abs{t-y_i}\).  Part~(i) of
\cref{lem:delta0-geometry} therefore makes each \(h_i\) affine or concave
quadratic between consecutive branch points.  In particular, none of the
\(h_i\) has a branch point in \((0,\alpha)\); none has one in
\((\alpha,\beta)\).

When \(\alpha<\beta\), we also have \(\alpha<t\), so
\(c(y_1)=\alpha\) forces \(y_1=t\pm\alpha\).  The coordinates \(y_2,y_3\)
lie outside \((t-\beta,t+\beta)\).  An interval of radius
\(\beta-\alpha\) around either \(t-\alpha\) or \(t+\alpha\) lies inside
that central interval, and \(\beta-\alpha\le t-\alpha\).  Thus part~(iii)
of \cref{lem:delta0-geometry} applies with either point as its centre.

These boundary estimates supply the endpoint derivatives needed for
\textup{(a)} and \textup{(b)}.

Because \(f\) and the \(h_i\) are polynomial on each adjacent branch and
continuous at its endpoints, each endpoint derivative used below equals the
corresponding limit of the ordinary derivative from within that branch.  We
write these limits explicitly.

\paragraph{Slope bounds at \(0\) and \(\alpha\).}
Suppose first that \(\alpha>0\).  Then every \(h_i\) is differentiable at
zero.  Differentiating \eqref{eq:triple-collected} for \(s>0\) and then letting
\(s\to0^+\) gives
\begin{align*}
 \lim_{s\to0^+}f'(s)={}
 &\Delta_0(y_2,y_3)\bigl[h_1(0)-\beta h_1'(0)\bigr]\\
 &+\Delta_0(y_1,y_3)\bigl[h_2(0)-\alpha h_2'(0)\bigr]\\
 &+\Delta_0(y_1,y_2)\bigl[h_3(0)-\alpha h_3'(0)\bigr].
\end{align*}
Concavity on \([-\alpha,0]\) gives
\(h_i(0)-\alpha h_i'(0)\ge h_i(-\alpha)\).  Hence
\begin{align*}
 \lim_{s\to0^+}f'(s)\ge{}
 &\Delta_0(y_2,y_3)
 \bigl[h_1(-\alpha)-(\beta-\alpha)h_1'(0)\bigr]+\Delta_0(y_1,y_3)h_2(-\alpha)
  +\Delta_0(y_1,y_2)h_3(-\alpha).
\end{align*}
This is nonnegative if \(\alpha=\beta\) or \(h_1'(0)\le0\).  In the only
remaining case, the piecewise formula in part~(i) forces
\(y_1=t-\sigma\alpha\).  Therefore
\(h_2(-\alpha)=\Delta_0(y_1,y_2)\),
\(h_3(-\alpha)=\Delta_0(y_1,y_3)\), and \(h_1'(0)\le y_1\).  Applying
part~(iii) of \cref{lem:delta0-geometry} at centre \(y_1\), with radius
\(\beta-\alpha\), yields
\[
 2\Delta_0(y_1,y_2)\Delta_0(y_1,y_3)
 \ge(\beta-\alpha)y_1\Delta_0(y_2,y_3).
\]
The two positive terms in the preceding lower bound therefore sum to
\(2\Delta_0(y_1,y_2)\Delta_0(y_1,y_3)\), which covers its only negative
term.  Thus \(\lim_{s\to0^+}f'(s)\ge0\).

At the other boundary, taking \(s\to\alpha^-\) in the derivative formula,
when \(\alpha>0\), gives
\begin{align*}
 \lim_{s\to\alpha^-}f'(s)={}
 &\Delta_0(y_1,y_3)h_2(\alpha)
  +\Delta_0(y_1,y_2)h_3(\alpha)
  +\Delta_0(y_2,y_3)h_1(\alpha)
 -(\beta-\alpha)\Delta_0(y_2,y_3)
   \lim_{s\to\alpha^-}h_1'(s).
\end{align*}
This is nonnegative when \(\alpha=\beta\).  When \(\alpha<\beta\),
taking \(s\to\alpha^+\) in the derivative formula, valid also for
\(\alpha=0\), gives
\begin{align*}
 \lim_{s\to\alpha^+}f'(s)={}
 &\Delta_0(y_1,y_3)h_2(\alpha)
  +\Delta_0(y_1,y_2)h_3(\alpha)
 -(\beta-\alpha)\Delta_0(y_2,y_3)
   \lim_{s\to\alpha^+}h_1'(s).
\end{align*}

Apply part~(iii) of \cref{lem:delta0-geometry} twice, with radius
\(\beta-\alpha\), first at centre \(y_1\) and then at centre
\(t+\sigma\alpha\):
\begin{align*}
 2\Delta_0(y_1,y_2)\Delta_0(y_1,y_3)
 &\ge(\beta-\alpha)y_1\Delta_0(y_2,y_3),\\
 2h_2(\alpha)h_3(\alpha)
 &\ge(\beta-\alpha)(t+\sigma\alpha)\Delta_0(y_2,y_3).
\end{align*}
Multiplying these inequalities bounds the product of the two left-hand
summands below.  Applying Arithmetic Mean–Geometric Mean inequality to those summands yields
\begin{align*}
 \Delta_0(y_1,y_3)h_2(\alpha)
  +\Delta_0(y_1,y_2)h_3(\alpha)\ge(\beta-\alpha)
 \sqrt{y_1(t+\sigma\alpha)}\,\Delta_0(y_2,y_3).
\end{align*}
For either applicable limit \(\lim_{s\to\alpha^-}h_1'(s)\) or
\(\lim_{s\to\alpha^+}h_1'(s)\), the piecewise formula in part~(i) gives the
following bound whenever the limit is positive.  If
\(y_1=t+\sigma\alpha\), it equals \(y_1\); if
\(y_1=t-\sigma\alpha\), it is at most \(t-\alpha\).  In either case it is at
most \(\sqrt{y_1(t+\sigma\alpha)}\).  Since \(h_1(\alpha)\ge0\), the two
limit formulas now imply
\[
 \lim_{s\to\alpha^+}f'(s)\ge0\quad(\alpha<\beta),
 \qquad
 \lim_{s\to\alpha^-}f'(s)\ge0\quad(\alpha>0).
\]

\paragraph{The interval \([0,\alpha]\).}
On \((0,\alpha)\), all three minima in \eqref{eq:triple-collected} equal
\(s\), and no \(h_i\) changes branch.  Hence
\[
 f'''(s)=3\Bigl[
 \Delta_0(y_2,y_3)h_1''(s)
 +\Delta_0(y_1,y_3)h_2''(s)
 +\Delta_0(y_1,y_2)h_3''(s)\Bigr]\le0.
\]
Thus \(f'\) is concave.  If \(\alpha>0\), its two endpoint limits are
nonnegative, so concavity gives \(f'\ge0\) on \((0,\alpha)\), and \(f\) is
nondecreasing on \([0,\alpha]\).  For \(\alpha=0\), the assertion is
vacuous.  In either case,
\[
 f(\alpha)\ge f(0)=0.
\]

\paragraph{The interval \([\alpha,\beta]\).}
This interval needs attention only when \(\alpha<\beta\).  No \(h_i\) changes
branch in its interior.  Up to terms independent of \(s\),
\[
 f(s)=-(\beta-\alpha)\Delta_0(y_2,y_3)h_1(s)
 +(s-\alpha)\Bigl[
 \Delta_0(y_1,y_3)h_2(s)+\Delta_0(y_1,y_2)h_3(s)\Bigr].
\]
Each \(h_i\) is affine or quadratic on this interval, so
\[
 f'''(s)=3\Bigl[
 \Delta_0(y_1,y_3)h_2''(s)
 +\Delta_0(y_1,y_2)h_3''(s)\Bigr]\le0.
\]
Therefore \(f'\) is concave.  Because
\(\lim_{s\to\alpha^+}f'(s)\ge0\), adjoining \(\alpha\) to the nonnegative
superlevel set of \(f'\) yields an interval.  Consequently \(f\) first rises
and may then fall, so its minimum on \([\alpha,\beta]\) occurs at an endpoint.
We already proved \(f(\alpha)\ge0\), and Case~1 gives \(f(\beta)\ge0\).  Thus
\(f(s)\ge0\) throughout.

The weak inequalities include tied service costs and coincident reports;
continuity covers the branch endpoints.  This proves
\cref{lem:cyclic-compensation}.

\section{Social cost approximation for
\texorpdfstring{\(\RAV_k\)}{RAV-k}}
\label{sec:approximation}

\begin{restatedmainthm}{thm:intro-approximation}{Social-cost approximation}
For every \(k \geq 2\) and every profile \(\mathbf{x}\),
\[
  \E\!\left[
    \SC\bigl(\mathbf{x},\RAVk{k}(\mathbf{x})\bigr)
  \right]
  \leq 4(k-1)\OPT_k(\mathbf{x}).
\]
In particular, the strategyproof-in-expectation mechanism \(\RAVk{3}\)
achieves an approximation factor of \(8\) for social cost.
\end{restatedmainthm}

The proof has three steps.  We condition on the anchor, group the outcomes
that differ only in which agent is left out, and compare each group with an
optimal \(k\)-facility solution.  We then average these bounds over the
uniformly chosen anchor.

\subsection{Conditioning on the anchor and grouping outcomes}

Fix a profile \(\mathbf{x}\) and an anchor agent \(a\in N\).  Condition on
\(I=a\), put \(p:=x_a\), and, for
\(S\in\binom{N\setminus\set{a}}{k-1}\), write
\[
  F_{a,S}:=\set{p}\cup\set{x_j:j\in S}.
\]
Whenever \(D_a^{(k)}(\mathbf{x})>0\), the definition of \(\RAVk{k}\) gives
\[
  \Pr(\mathcal{S}=S\mid I=a)
  =\frac{\Delta_p(\mathbf{x}_S)}{D_a^{(k)}(\mathbf{x})}.
\]
Only agents outside \(S\cup\set{a}\) can have positive cost.  Each such pair
\((S,i)\) corresponds to
\[
  (S,i)\longleftrightarrow(U=S\cup\set{i},i),
\]
where \(U\in\binom{N\setminus\set{a}}{k}\) and \(i\in U\); conversely,
\(S=U\setminus\set{i}\).  Hence
\begin{align}
 D_a^{(k)}(\mathbf{x})\,
 \E\!\left[
   \SC\bigl(\mathbf{x},\RAVk{k}(\mathbf{x})\bigr)
   \given I=a
 \right]\notag
 &=
 \sum_{S\in\binom{N\setminus\set{a}}{k-1}}
 \Delta_p(\mathbf{x}_S)
 \sum_{i\in N\setminus(S\cup\set{a})}
 d(x_i,F_{a,S})\notag\\
 &=
 \sum_{U\in\binom{N\setminus\set{a}}{k}}
 \sum_{i\in U}
 \Delta_p\bigl(\mathbf{x}_{U\setminus\set{i}}\bigr)
 d\bigl(x_i,F_{a,U\setminus\set{i}}\bigr).
\label{eq:ravk-grouped-numerator}
\end{align}
For a fixed \(U\), the inner sum covers all \(k\) ways to leave one member of
\(U\) unselected.  For example, when \(k=3\) and \(U=\set{u,v,w}\), its three
terms leave out \(u\), \(v\), and \(w\), respectively.  The next bound applies
to their sum; it need not hold separately for each choice.

There is a simple reason to keep these \(k\) terms together.  The anchor and
the \(k\) agents in \(U\) give \(k+1\) agents, but an optimal solution uses at
most \(k\) facilities.  After assigning agents to optimal facilities as
described below, two of these agents must therefore be assigned to the same
facility.  In one dimension, the agents assigned to a facility form a
contiguous block in report order.  Hence some two agents that are adjacent
among the \(k+1\) reports in \(\set{a}\cup U\) are assigned to the same
facility.  The proof below uses that adjacent pair to bound the complete sum
for \(U\).

Fix an optimal facility set
\[
  F^*=\set{f_1,\ldots,f_\ell},
  \qquad f_1<\cdots<f_\ell,
  \qquad \ell\le k.
\]
Assign every agent to a nearest facility in \(F^*\), breaking ties toward the
leftmost facility.  For \(f\in F^*\), let
\[
  N_f:=\set{i\in N:i\text{ is assigned to }f},
  \qquad
  C_f^*:=\sum_{i\in N_f}\abs{x_i-f}.
\]
Then \(\sum_{f\in F^*}C_f^*=\OPT_k(\mathbf{x})\).  Each \(N_f\) is contiguous
in report order because the points nearest to a fixed facility form an
interval on the line.  In particular, co-located agents have the same
assignment.

\subsection{Bounding one group}

For an anchor \(a\in N_f\), associate with every agent \(i\) the quantity
\[
  b_i^a:=
  \begin{cases}
    (2k-2)\abs{x_i-x_a},&i\in N_f,\\
    (2k-1)d(x_i,F^*),&i\notin N_f.
  \end{cases}
\]
These quantities are used only in the following comparison.  Notice that
\(b_a^a=0\).  The factor \(2k-1\) will be used when two neighboring agents
belong to another optimal cluster; the factor \(2k-2\) will be used when a
neighbor belongs to the anchor's cluster.

\begin{lemma}[Bound for one group]
\label{lem:ravk-local-omission}
If \(a\in N_f\), then every
\(U\in\binom{N\setminus\set{a}}{k}\) satisfies
\begin{align}
 \sum_{i\in U}
 \Delta_p\bigl(\mathbf{x}_{U\setminus\set{i}}\bigr)
 d\bigl(x_i,F_{a,U\setminus\set{i}}\bigr)
 \le
 \sum_{i\in U}
 \Delta_p\bigl(\mathbf{x}_{U\setminus\set{i}}\bigr)b_i^a.
\label{eq:ravk-local-omission}
\end{align}
\end{lemma}

The left-hand side is the contribution of \(U\) to
\eqref{eq:ravk-grouped-numerator}.  The right-hand side retains the volume
weight assigned to each outcome and replaces the left-out agent's service
distance by \(b_i^a\).  Thus \eqref{eq:ravk-local-omission} compares the sums
over all \(k\) possible left-out agents; it is not a separate inequality for
each \(i\).

\begin{proof}
Write \(L\) and \(R\) for the left- and right-hand sides of
\eqref{eq:ravk-local-omission}.  First suppose that the \(k+1\) coordinates
of the anchor and the agents in \(U\) are not all distinct.  For each
\(i\in U\), either the coordinates that remain after \(i\) is left out still
contain a repetition, so one adjacent gap and hence the volume is zero, or
\(x_i\) coincides with a remaining coordinate, so a facility remains open at
\(x_i\).  Thus every term in \(L\) is zero, while \(R\ge0\).

Now suppose that all \(k+1\) coordinates are distinct.  Write
\[
\begin{gathered}
  z_0<z_1<\cdots<z_k,\qquad p=z_q,\qquad
  g_j:=z_j-z_{j-1}\quad(1\le j\le k),\\
  G:=\prod_{j=1}^k g_j>0,\qquad
  \rho_j:=\frac{\min\set{g_j,g_{j+1}}}
                 {\max\set{g_j,g_{j+1}}}
  \quad(1\le j\le k-1).
\end{gathered}
\]
Let \(i_j\) be the agent whose report is \(z_j\), so \(i_q=a\).  Leaving out
an endpoint removes its incident gap from the volume product.  Leaving out an
interior report \(z_j\) replaces the two factors \(g_jg_{j+1}\) by their sum
\(g_j+g_{j+1}\).  The corresponding volume and service distance are
\[
\begin{array}{c|c|c}
  j\ne q
  &\Delta_p(\mathbf{x}_{U\setminus\set{i_j}})/G
  &d(x_{i_j},F_{a,U\setminus\set{i_j}})\\ \hline
  0&1/g_1&g_1\\
  1\le j\le k-1&1/g_j+1/g_{j+1}&\min\set{g_j,g_{j+1}}\\
  k&1/g_k&g_k
\end{array}
\]
In each endpoint row, the volume divided by \(G\) is \(1/g\) and the service
distance is \(g\), so their product is one.  In the interior row, the product
of the two displayed entries is \(1+\rho_j\), by the definition of
\(\rho_j\).
Thus each of the \(k\) ways to leave out a non-anchor agent contributes one
unit to \(L/G\), and leaving out an interior agent contributes the additional
ratio \(\rho_j\).  The anchor at position \(q\) is never left out, so
\begin{equation}\label{eq:ravk-normalized-omission}
  \frac{L}{G}
  =k+\sum_{\substack{1\le j\le k-1\\j\ne q}}\rho_j.
\end{equation}
Since every \(\rho_j\le1\), this is at most \(2k-1\).  If the anchor is
interior, one interior position is absent from the sum, so \(L/G\le2k-2\).

The middle column of the table also gives the coefficient of each \(b_{i_j}^a\)
in \(R/G\).  Since \(b_{i_q}^a=b_a^a=0\), we may include the anchor's zero term
and collect the two \(b\)-terms associated with the endpoints of each gap:
\[
\begin{aligned}
  \frac{R}{G}
  &=
  \frac{b_{i_0}^a}{g_1}
  +\sum_{j=1}^{k-1}b_{i_j}^a
    \left(\frac1{g_j}+\frac1{g_{j+1}}\right)
  +\frac{b_{i_k}^a}{g_k}\\
  &=\sum_{h=1}^k
    \frac{b_{i_{h-1}}^a+b_{i_h}^a}{g_h}.
\end{aligned}
\]
This is the gap-by-gap form of the right-hand side: each summand
uses one gap and the two agents at its endpoints.  Indeed, the term
\(b_{i_j}^a/g_j\) comes from the gap immediately to the left of \(z_j\), and
\(b_{i_j}^a/g_{j+1}\) comes from the gap immediately to its right.  Collecting
equal gap denominators produces the last sum.  Every summand is nonnegative.

It remains to compare this nonnegative sum with the bound for \(L/G\).  The
\(k+1\) agents are assigned to at most \(k\) optimal facilities.  Suppose first
that a cluster \(N_{\widehat f}\ne N_f\) contains at least two of them.
Because that cluster is contiguous, two of its agents occupy adjacent
positions, say \(z_{h-1}\) and \(z_h\).  Neither is the anchor, and the
triangle inequality gives
\[
  \frac{R}{G}\ge
  (2k-1)
  \frac{\abs{z_{h-1}-\widehat f}+\abs{z_h-\widehat f}}
       {z_h-z_{h-1}}
  \ge2k-1\ge\frac{L}{G}.
\]

Otherwise, every cluster other than \(N_f\) contains at most one of these
agents.  Since there are at most \(k-1\) such clusters, \(N_f\) contains the
anchor and another agent.  Contiguity implies that an immediate neighbor
\(z_j\) of the anchor also lies in \(N_f\).  If \(0<q<k\), then
\[
  \frac{R}{G}\ge
  \frac{b_{i_j}^a}{\abs{z_j-p}}
  =2k-2\ge\frac{L}{G}.
\]
If \(q=0\), then \(z_1\in N_f\), and
\[
  \frac{L}{G}
  \le2k-2+\rho_1
  \le(2k-2)\left(1+\frac{g_1}{g_2}\right)
  =b_{i_1}^a\left(\frac1{g_1}+\frac1{g_2}\right)
  \le\frac{R}{G}.
\]
Here the second inequality uses \(2k-2\ge1\) and
\(\rho_1\le g_1/g_2\).  The case \(q=k\) follows by reflecting the line.
When \(k=2\), there are no other ratios to bound, so the same endpoint
calculation applies.  Hence \(L\le R\) in every case.
\end{proof}

\subsection{Averaging over the anchor}

\begin{proof}[Proof of \cref{thm:intro-approximation}]
Fix \(k\ge2\).  If the profile has fewer than \(k\) distinct locations, the
fallback rule opens at every distinct location and has zero social cost.  If
\(n=k\), every reported location is opened and the social cost is again zero.
We may therefore assume that \(n\ge k+1\) and that the profile has at least
\(k\) distinct locations.  Then \(D_a^{(k)}(\mathbf{x})>0\) for every anchor
\(a\).

Fix an anchor \(a\), and let \(N_f\) be its optimal cluster.  Apply
\cref{lem:ravk-local-omission} to each \(U\) in
\eqref{eq:ravk-grouped-numerator}, then reverse the earlier bijection:
\begin{align*}
 D_a^{(k)}(\mathbf{x})
 \E\!\left[
   \SC\bigl(\mathbf{x},\RAVk{k}(\mathbf{x})\bigr)
   \given I=a
 \right]
 &\le
 \sum_{U\in\binom{N\setminus\set{a}}{k}}
 \sum_{i\in U}
 \Delta_p\bigl(\mathbf{x}_{U\setminus\set{i}}\bigr)b_i^a\\
 &=
 \sum_{S\in\binom{N\setminus\set{a}}{k-1}}
 \Delta_p(\mathbf{x}_S)
 \sum_{i\in N\setminus(S\cup\set{a})}b_i^a\\
 &\le
 D_a^{(k)}(\mathbf{x})
 \sum_{i\in N\setminus\set{a}}b_i^a.
\end{align*}
The last inequality holds because every \(b_i^a\) is nonnegative.  After
dividing by \(D_a^{(k)}(\mathbf{x})\) and substituting the definition of
\(b_i^a\), we obtain
\[
 \E\!\left[
   \SC\bigl(\mathbf{x},\RAVk{k}(\mathbf{x})\bigr)
   \given I=a
 \right]
 \le
 (2k-1)\bigl(\OPT_k(\mathbf{x})-C_f^*\bigr)
 +(2k-2)\sum_{j\in N_f}\abs{x_j-x_a}.
\]
The quantity \(\OPT_k(\mathbf{x})-C_f^*\) is the total optimal cost of the
clusters other than the anchor's cluster.  The second sum ranges only over
agents in the anchor's cluster.

For \(f\in F^*\), let \(n_f:=\abs{N_f}\).  When we average over all possible
anchors in \(N_f\), the second term above contains the distance between every
ordered pair of agents in that cluster.  For each pair, the triangle
inequality bounds their distance by the sum of their distances to \(f\).
Summing this inequality over all pairs gives
\begin{equation}\label{eq:optimal-cluster-pairwise-bound}
  \sum_{a\in N_f}\sum_{j\in N_f}\abs{x_j-x_a}
  \le
  \sum_{a\in N_f}\sum_{j\in N_f}
  \bigl(\abs{x_j-f}+\abs{x_a-f}\bigr)
  =2n_fC_f^*.
\end{equation}
The anchor is uniform on \(N\).  Averaging the conditional bound above,
grouping anchors by their optimal cluster, and using
\eqref{eq:optimal-cluster-pairwise-bound} yields
\begin{align*}
 \E\!\left[
   \SC\bigl(\mathbf{x},\RAVk{k}(\mathbf{x})\bigr)
 \right]
 &\le
 (2k-1)\OPT_k(\mathbf{x})
 -\frac{2k-1}{n}\sum_{f\in F^*}n_fC_f^*
 +\frac{2(2k-2)}{n}\sum_{f\in F^*}n_fC_f^*\\
 &=
 (2k-1)\OPT_k(\mathbf{x})
 +\frac{2k-3}{n}\sum_{f\in F^*}n_fC_f^*\\
 &\le
 \bigl((2k-1)+(2k-3)\bigr)\OPT_k(\mathbf{x})\\
 &=4(k-1)\OPT_k(\mathbf{x}).
\end{align*}
The first line uses \(\sum_f n_f=n\); the last inequality uses
\(n_f\le n\) and \(\sum_f C_f^*=\OPT_k(\mathbf{x})\).  The coefficient
\(2k-3\) is \(2(2k-2)-(2k-1)\): the within-cluster term contributes the first
quantity, while excluding the anchor's cluster from the other term subtracts
the second.  For \(k=3\), the factor is \(8\).
\end{proof}

\section{The truthfulness boundary for Random-Anchor Volume}
\label{sec:general-incentive-failure}

As in the proof of \cref{thm:rav3-sp}, it is useful to begin by conditioning
on the identity of the anchor.  Once the anchor is fixed, only the
second-stage volume sampling remains, so the incentive comparison is much
simpler.  For the positive result, it was enough to prove the truthfulness
inequality for every possible anchor and then average.  The negative direction
has an important asymmetry: a profitable deviation for one fixed anchor does
not by itself remain profitable after averaging over the other possible
anchors.

We therefore separate the conditional calculation from the passage back to
the mechanism as defined.  First, we exhibit a strict conditional deviation
for \(\RAVk{4}\) with the anchor at \(0\).  We then add sufficiently many
agents at the anchor coordinate.  Whenever one of those agents is selected as
the anchor, the same profitable conditional deviation reappears, while the
combined contribution from all other possible anchors remains bounded.  This
produces a manipulation of the mechanism with a uniformly chosen anchor.  The
passage is needed already for \(k=4\).  To obtain the result for \(k>4\), we
first add sufficiently distant reports while preserving the conditional
deviation and then apply the same passage to the full mechanism.

The idea behind the four-facility example is simple.  The deviating agent has
many other agents at her true location.  Moving a small distance creates many
new positive-volume selections that contain both her alternative report and a
facility at her true location, so all of those outcomes still cost her zero.
The multiplicities of the negative reports are chosen so that each negative
group has the same multiplicity-weighted distance from the anchor.  This
balance makes the added zero-cost volume reduce the probability of costly
outcomes by more than the small direct cost created by the move.

\subsection{A four-facility deviation conditional on the anchor}

\begin{proposition}[A profitable report conditional on the anchor]
\label{prop:rav4-fixed-anchor-violation}
Conditional on a distinguished agent at \(p=0\) being selected as the anchor,
there is a finite rational instance in which an agent with true location
\(t=1\) strictly lowers her expected cost under \(\RAVk{4}\) by reporting
\(r=1001/1000\).  The total adjacent-gap volume is positive under both reports.
\end{proposition}

\begin{proof}
Fix the anchor at \(p=0\), let the agent's true location be \(t=1\), and put
\[
  \varepsilon=\frac{1}{1000},
  \qquad r=1+\varepsilon=\frac{1001}{1000}.
\]
Besides this agent, use the following reports.

\begin{center}
\begin{tabular}{@{}lll@{}}
\toprule
Role & Reported location & Number of agents \\
\midrule
Other agents at the true location & \(1\) & \(100\) \\
Right-hand agent & \(2\) & \(1\) \\
Negative group \(j\), \(1\le j\le100\)
  & \(-1/(100\cdot2^j)\) & \(2^j\) \\
\bottomrule
\end{tabular}
\end{center}

No non-anchor report is at \(0\), under either the truthful report or the
alternative report.

All coordinates are rational and the population is finite.  When the agent
reports truthfully, every selected triple containing her opens a facility at
\(1\), so those outcomes cost her zero.  After she reports \(1+\varepsilon\),
any selected triple that also contains one of the other \(100\) reports at
\(1\) still costs her zero.  The move makes many such triples have positive
volume: at the truthful profile, selecting two agents at coordinate \(1\)
would instead create a zero gap.

\paragraph{Weights contributed by the negative groups.}
Group \(j\) contains \(2^j\) agents, each at distance \(1/(100\cdot2^j)\)
from the anchor.  Thus every group has the same multiplicity-weighted distance
from the anchor,
\[
  2^j\frac{1}{100\cdot2^j}=\frac{1}{100},
\]
and the total over all \(100\) groups is \(1\).

Let \(W_2\) be the total adjacent-gap volume of all choices of one agent from
each of two distinct negative groups.  Define \(W_3\) analogously for three
distinct negative groups.  Multiplicity is included in both totals.  Choosing
two agents from the same group gives volume zero because their reports
coincide.  Thus \(W_2\) and \(W_3\) are simply sums of the volumes
\(\Delta_0\) defined in \cref{eq:delta-k-definition}, restricted to choices
from two or three distinct negative groups.  For \(i<j\), the contribution of
groups \(i\) and \(j\) to \(W_2\) is
\[
  2^i2^j\frac{1}{100\cdot2^j}
  \left(\frac{1}{100\cdot2^i}-\frac{1}{100\cdot2^j}\right)
  =\frac{1-2^{i-j}}{100^2}.
\]
The corresponding three-group calculation gives
\begin{align*}
 W_2
 &=\frac{1}{100^2}\sum_{1\le i<j\le100}\bigl(1-2^{i-j}\bigr),\\
 W_3
 &=\frac{1}{100^3}\sum_{1\le i<j<h\le100}
   \bigl(1-2^{i-j}\bigr)\bigl(1-2^{j-h}\bigr).
\end{align*}
Summing the geometric series gives the exact values
\begin{align*}
 W_2
 &=\frac12-\frac{3}{200}+\frac{2}{100^2}
   -\frac{1}{100^2 2^{99}},\\
 W_3
 &=\frac16-\frac{3}{200}+\frac{19}{3\cdot100^2}
   -\frac{12}{100^3}+\frac{212}{2^{100}100^3}.
\end{align*}
In particular,
\begin{equation}
  0.485<W_2<0.486,
  \qquad
  0.152<W_3<0.153.
  \label{eq:rav4-left-weight-bounds}
\end{equation}

\paragraph{The conditional cost comparison.}
We reuse the notation \(Z,A,U(z),B(r),C_t,C_r\) defined in
\cref{sec:fixed-anchor}, now with three reports selected after the anchor.
Thus \(Z\) is the total volume of triples that omit the deviating agent, while
\(A\) is their cost-weighted total.  Likewise, \(U(z)\) is the total volume of
triples containing the agent when she reports \(z\), and \(B(r)\) is their
cost-weighted total under the alternative report \(r\).  The positive-volume
triples that omit her are as follows; a negative report means a report from
one of the groups above.

\begin{center}
\begin{tabularx}{\textwidth}{@{}Xcc@{}}
\toprule
Selected reports besides the anchor & Total volume & Agent's cost \\
\midrule
One of the \(100\) other reports at \(1\), the report at \(2\), and one negative
  & \(100\) & \(0\) \\
One of the \(100\) other reports at \(1\) and two negatives
  & \(100W_2\) & \(0\) \\
The report at \(2\) and two negatives
  & \(2W_2\) & \(1\) \\
Three negatives
  & \(W_3\) & \(1\) \\
\bottomrule
\end{tabularx}
\end{center}

Consequently,
\[
  Z=100+102W_2+W_3,
  \qquad
  A=2W_2+W_3.
\]
If the agent reports \(t=1\), the positive-volume triples containing her use
either the report at \(2\) and one negative, or two negatives.  Their total
volume is therefore
\[
  U(t)=1+W_2,
\]
and all of them give the agent cost zero.

If she reports \(r=1+\varepsilon\), the positive-volume triples containing her
are the following.

\begin{center}
\begin{tabularx}{\textwidth}{@{}Xcc@{}}
\toprule
Other two selected reports & Total volume & Agent's cost \\
\midrule
One of the \(100\) other reports at \(1\) and the report at \(2\)
  & \(100\varepsilon(1-\varepsilon)\) & \(0\) \\
One of the \(100\) other reports at \(1\) and one negative
  & \(100\varepsilon\) & \(0\) \\
The report at \(2\) and one negative
  & \(1-\varepsilon^2\) & \(\varepsilon\) \\
Two negatives
  & \((1+\varepsilon)W_2\) & \(\varepsilon\) \\
\bottomrule
\end{tabularx}
\end{center}

Hence
\begin{align*}
 U(r)
 &=100\varepsilon(1-\varepsilon)+100\varepsilon
   +(1-\varepsilon^2)+(1+\varepsilon)W_2,\\
 B(r)
 &=\varepsilon(1-\varepsilon^2)
   +\varepsilon(1+\varepsilon)W_2.
\end{align*}
As in \cref{sec:fixed-anchor}, the two conditional expected costs are
\[
  C_t=\frac{A}{Z+U(t)},
  \qquad
  C_r=\frac{A+B(r)}{Z+U(r)}.
\]
Both denominators are positive because \(Z\ge100\).  Therefore the alternative
report is profitable exactly when the numerator in
\eqref{eq:fixed-anchor-numerator} is negative.  Substitution gives
\begin{align}
&\frac{B(r)\bigl(Z+U(t)\bigr)-A\bigl(U(r)-U(t)\bigr)}{\varepsilon}
\notag\\
&\quad=
 \bigl[101W_2^2-196W_2-199W_3+101\bigr]
\notag\\
&\qquad
 +\varepsilon\bigl[103W_2^2+W_2W_3+303W_2+101W_3\bigr]
\notag\\
&\qquad
 -\varepsilon^2\bigl[103W_2+W_3+101\bigr].
\label{eq:rav4-explicit-cost-comparison}
\end{align}
On the rectangle in \eqref{eq:rav4-left-weight-bounds}, the first bracket is
decreasing in both \(W_2\) and \(W_3\), whereas the second bracket is
increasing in both variables.  The bounds therefore give
\[
  101W_2^2-196W_2-199W_3+101<-0.55
\]
and
\[
  103W_2^2+W_2W_3+303W_2+101W_3<188.
\]
The final bracket in \eqref{eq:rav4-explicit-cost-comparison} is positive.
Since \(\varepsilon=10^{-3}\),
\[
 \frac{B(r)\bigl(Z+U(t)\bigr)-A\bigl(U(r)-U(t)\bigr)}{\varepsilon}
 <-0.55+10^{-3}\cdot188
 =-0.362<0.
\]
Thus \(C_r<C_t\), as claimed.
\end{proof}

\subsection{Returning to the uniformly chosen anchor}

The proposition establishes a profitable deviation only after the anchor
identity has been fixed.  We now pass to the mechanism as defined, in which
the anchor agent is chosen uniformly.

\begin{lemma}[From a fixed anchor to the full mechanism]
\label{lem:conditioned-to-full-mechanism}
Fix \(k\ge2\).  Let \(a\ne i\) be a distinguished agent reporting \(p\), and
suppose agent \(i\), whose true location is \(t\), can instead report \(r\).
Conditional on \(a\) being selected as the anchor, assume that the finite
\(\RAVk{k}\) instance has no non-anchor report at \(p\) under either report of
agent \(i\) and has positive total volume under both reports.  Let \(C_t\) and
\(C_r\) be agent \(i\)'s conditional expected costs when she reports \(t\) and
\(r\), respectively.  If \(C_r<C_t\), then adding finitely many agents who
report \(p\) produces a finite profile on which the same report \(r\) is
profitable for the mechanism with a uniformly chosen anchor.
\end{lemma}

\begin{proof}
Let \(y_1,\ldots,y_m\) be the reports of the non-anchor agents other than the
deviating agent, and write
\[
  \delta=C_r-C_t<0.
\]
The number of possible anchor identities away from \(p\) will remain fixed,
whereas every additional identity at \(p\) will contribute another copy of
the same negative conditional difference \(\delta\).

Form a full profile with a total of \(q\) agents at \(p\), including the
original anchor agent.  Thus we add \(q-1\) new agents at \(p\).  If one of the
\(q\) agents at \(p\) is chosen as the anchor, any selected set containing
another agent at \(p\) has zero adjacent-gap volume.  The remaining
positive-volume selections are exactly those of the original conditional
instance.  Each of the \(q\) anchor identities at \(p\) therefore contributes
the same negative cost difference \(\delta\).

Now consider the other possible anchors.  If the agent at \(y_j\) is chosen,
a facility at \(y_j\) is open under both reports.  The conditional cost under
the alternative report is therefore at most \(\abs{t-y_j}\), so the
alternative-report cost minus the truthful cost is also at most this amount.
If the deviating agent is the anchor, the corresponding cost difference is at
most \(\abs{t-r}\).  These bounds remain valid if the fallback rule is used.
Thus the sum of the conditional cost differences over every anchor identity
not at \(p\) is at most
\[
  \abs{t-r}+\sum_{j=1}^m\abs{t-y_j},
\]
which is finite and does not depend on \(q\).

Choose \(q\) large enough that
\[
  q\delta+\abs{t-r}+\sum_{j=1}^m\abs{t-y_j}<0.
\]
The sum of all conditional cost differences is at most the left-hand side
above and is therefore negative.  Since the anchor is uniform, the
unconditional cost difference is that actual sum divided by the total number
of agents.  Hence the report \(r\) is strictly profitable.
\end{proof}

Applying \cref{lem:conditioned-to-full-mechanism} directly to
\cref{prop:rav4-fixed-anchor-violation} already gives a finite profile on
which \(\RAVk{4}\), with its uniformly chosen anchor, is manipulable.  The
next argument is needed only for mechanisms with more than four facilities.

\subsection{Extending the manipulation beyond four facilities}

We next show that the conditional example survives when one more facility is
requested.  The only change is the addition of one agent far to the right.

\begin{lemma}[Adding a distant report preserves the deviation]
\label{lem:add-distant-report}
Fix \(k\ge2\) and condition on a distinguished agent at \(p\) being selected
as the anchor.  Suppose an agent with true location \(t\) can instead report
\(r\), and that, for a finite collection of other agents, the conditional
\(\RAVk{k}\) lotteries have positive total volume under both reports and
satisfy \(C_r<C_t\).  For every sufficiently large coordinate \(L\), adding
one additional agent whose report is fixed at \(L\) makes the same report
\(r\) profitable for the conditional \(\RAVk{k+1}\) lottery.  If all original
coordinates are rational, \(L\) may be chosen to be an integer.
\end{lemma}

\begin{proof}
For \(z\in\{t,r\}\), let \(D_z>0\) be the total adjacent-gap volume of the
original conditional lottery, so its cost-weighted total is \(D_zC_z\).
Choose \(L\) to the right of the anchor, the two possible reports, and every
other report.  Also choose it far enough that opening a facility at \(L\)
cannot improve the agent's service beyond the facility already open at \(p\).

Consider a set of \(k-1\) reports that could be selected in the original
\(\RAVk{k}\) lottery.  Appending the new rightmost report at \(L\) leaves all
old consecutive gaps unchanged and adds one final gap.  This final gap is
\(L\) minus the old rightmost selected coordinate.  The service cost at the
true location is unchanged because the new facility is too far away.

There are only finitely many old selected sets.  The sets considered by
\(\RAVk{k+1}\) that do not contain the new report contribute quantities
independent of \(L\).  It follows that, for constants
\(\alpha_z,\beta_z\) independent of \(L\), the new conditional expected cost
has the form
\[
  C_z^{+}(L)
  =\frac{L D_zC_z+\beta_z}{L D_z+\alpha_z}.
\]
Consequently,
\[
  \lim_{L\to+\infty}C_z^{+}(L)=C_z.
\]
Since \(C_r<C_t\) is a strict inequality, we also have
\(C_r^{+}(L)<C_t^{+}(L)\) for every sufficiently large finite \(L\).
With rational original coordinates, any sufficiently large integer \(L\)
gives a finite rational instance.
\end{proof}

\begin{theorem}[Failure of \RAV strategyproofness beyond three facilities]
\label{thm:general-nonsp}
For every integer \(k\ge4\), there is a finite real-line profile, an agent
with true location \(t\), and a misreport \(r\) such that
\[
 \E\!\left[d\bigl(t,\RAVk{k}(r,\mathbf{x}_{-i})\bigr)\right]
 <
 \E\!\left[d\bigl(t,\RAVk{k}(t,\mathbf{x}_{-i})\bigr)\right].
\]
Consequently, \(\RAVk{k}\) is not strategyproof in expectation for any
\(k\ge4\).
\end{theorem}

\begin{proof}
For \(k=4\), apply \cref{lem:conditioned-to-full-mechanism} directly to the
strict conditional deviation in
\cref{prop:rav4-fixed-anchor-violation}.  This gives a finite profile on which
the mechanism with its uniformly chosen anchor is manipulable.

Now let \(k>4\).  Starting from the same conditional four-facility deviation,
apply \cref{lem:add-distant-report} exactly \(k-4\) times, placing each new
agent to the right of every coordinate already present.  The resulting
conditional \(\RAVk{k}\) instance retains the strict deviation and has no
additional report at the anchor coordinate.  A final application of
\cref{lem:conditioned-to-full-mechanism} gives a finite manipulation of the
mechanism with its uniformly chosen anchor.
\end{proof}

\begin{restatedmainthm}{thm:intro-boundary}{Exact \RAV truthfulness boundary}
For $k\in\{1,2,3\}$, the mechanism \RAVk{k} is
strategyproof in expectation.  In contrast, for every $k\ge4$, there exists a
profile under which some agent can strictly lower her expected distance by
misreporting.
\end{restatedmainthm}

\begin{proof}
For \(k=1\), the mechanism is random dictatorship: if the potentially deviating agent
is chosen, truth gives cost zero, and otherwise its report does not affect the
outcome.

For \(k=2\), once the anchor is \(p\), a candidate report \(x_j\) receives
weight \(\abs{x_j-p}\).  Thus the second facility-hosting agent is chosen with
probability proportional to its distance from the uniformly chosen anchor,
exactly as in the Proportional Mechanism of Lu, Sun, Wang, and Zhu
\cite{LuEtAl2010}.  Their strategyproofness result therefore applies to
\(\RAVk{2}\); when all reports coincide, both rules simply open at that common
location.

The case \(k=3\) is \cref{thm:rav3-sp}.  Failure for every \(k\ge4\) is
\cref{thm:general-nonsp}.
\end{proof}

\section{Discussion}\label{sec:conclusion}

The Random-Anchor Volume mechanism generalizes the two-facility Proportional
Mechanism of Lu, Sun, Wang, and Zhu~\cite{LuEtAl2010}. For every \(k\geq 2\),
it achieves a \(4(k-1)\)-approximation to the optimal social cost. Crucially,
when \(k=3\), this generalization satisfies strategyproofness in expectation.
It therefore yields the first strategyproof-in-expectation mechanism for three
facilities that achieves a constant-factor approximation to the optimal social
cost.

\bigskip

\noindent\textbf{Global $k$-tuple mechanism.}
While preparing this manuscript, we became aware of a recent paper by Ma and
Peng~\cite{MaPeng2026}, who introduce a new mechanism for \(k=2\), called the
\emph{Global Pair} mechanism. This raises the natural question of whether the
ideas underlying \(\RAV\) can be used to extend their mechanism to an arbitrary
number of facilities.

We define such an extension, which we call the \emph{Global \(k\)-tuple}($\mathrm{Glob}_k$)
mechanism. For each set \(T\in\binom{N}{k}\) of \(k\) distinct agents, arrange
their reported locations, allowing repetitions, in nondecreasing order:
\[
    z_1\leq z_2\leq\cdots\leq z_k.
\]
Assign \(T\) the adjacent-gap weight
\[
    \Delta(T)
    :=
    \prod_{j=1}^{k-1}(z_{j+1}-z_j).
\]
The mechanism then selects a set \(T\in\binom{N}{k}\) with probability
\[
    \Pr(\mathcal{T}=T)
    =
    \frac{\Delta(T)}
         {\displaystyle\sum_{S\in\binom{N}{k}}\Delta(S)}
\]
and opens facilities at the reported locations of the agents in \(T\). If \(\sum_{S\in\binom{N}{k}}\Delta(S)=0\), or equivalently, if the profile contains fewer than
\(k\) distinct reported locations, the mechanism opens a facility at every distinct reported location.

Using an argument similar to that in the proof of
\Cref{thm:intro-approximation}, one can show that \(\mathrm{Glob}_k\)
achieves a \(2k\)-approximation to the optimal social cost for every
\(k\geq 2\). It would be interesting to investigate whether the approach of
Ma and Peng~\cite{MaPeng2026} for \(k=2\), which randomizes between the
Global Pair and Proportional mechanisms to obtain an improved social-cost
approximation, has an analogue for larger \(k\). In particular, one could ask
whether randomizing between \(\mathrm{Glob}_k\) and \(\RAV_k\) yields a
better approximation guarantee than either mechanism alone. This question is
already interesting for \(k=3\).

We were also curious whether \(\mathrm{Glob}_3\) is
strategyproof in expectation. When we posed this question to ChatGPT, it
surprisingly produced a purported certificate suggesting that the mechanism
is indeed strategyproof in expectation. However, because these alternative mechanisms lie beyond the scope of the present paper, we did not pursue this further.
We leave determining
whether \(\mathrm{Glob}_3\) is strategyproof in expectation—and, if so,
obtaining a human-verifiable proof—as an interesting direction for future
work. 

\bigskip

\noindent\textbf{Extending beyond $k\geq 4$.} The most intriguing future direction is to obtain a truthful in expectation mechanism with constant social cost approximation bound for $k\geq 4$. 
One possible approach is to modify the
adjacent-gap volume used by \(\RAV\). In particular, one could replace it by a more general weight of the form
\[
\widetilde{\Delta}_p(\mathbf{x}_S)=\phi(z_0,z_1,...,z_{k-1})
\]
where $z_0,z_1,...,z_{k-1}$ are the locations of the anchor \(p\) and the reported locations in
\(\mathbf{x}_S\), arranged in nondecreasing order. It would be interesting to
determine whether an appropriate choice of \(\phi\) can restore
strategyproofness in expectation for $k\geq 4$ while retaining a constant-factor
approximation guarantee. We leave this as an open direction for future work.

\end{document}